\documentclass[11pt]{article}
\pdfoutput = 1

\usepackage[utf8]{inputenc}        
\usepackage[T1]{fontenc}
\usepackage{graphicx}			
\usepackage{color}				
\usepackage{amsmath, amssymb, amsthm}
\usepackage{thmtools,thm-restate}
\usepackage{tabularx}
\usepackage{multirow}
\usepackage{enumitem}
\usepackage[page, toc]{appendix}
\usepackage{tcolorbox}
\usepackage{authblk}
\usepackage{tikz}
\usepackage{adjustbox}
\usepackage{diagbox}

\usepackage{hyperref}

\usepackage{fullpage}

\newtheorem{theorem}{Theorem}[subsection]
\newtheorem{proposition}[theorem]{Proposition}
\newtheorem{conjecture}[theorem]{Conjecture}
\newtheorem{corollary}[theorem]{Corollary}
\newtheorem{lemma}[theorem]{Lemma}

\theoremstyle{definition}
\newtheorem{definition}[theorem]{Definition}
\newtheorem*{remark}{Remark}
\newtheorem*{example}{Example}

\renewcommand{\P}{\ensuremath{\mathcal P}}
\newcommand{\N}{\ensuremath{\mathcal N}}
\newcommand{\G}{\ensuremath{\mathcal G}}
\newcommand{\R}{\ensuremath{\mathcal R}}
\renewcommand{\iff}{\Leftrightarrow}

\providecommand{\keywords}[1]{\textbf{Keywords:} #1}

\newcommand{\swComp}{\circledcirc}
\newcommand{\defemph}{\textbf}
\newcommand{\Rthl}{\textsc{Zeruclid}}
\newcommand{\Subs}{\textsc{Subtraction}}
\DeclareMathOperator{\mex}{mex}

\definecolor{gameColor}{RGB}{255, 128, 0}
\definecolor{gameBack}{RGB}{255, 217, 179}
\definecolor{gameTitleBack}{RGB}{255, 245, 230}

\title{The switch operators and push-the-button games: a sequential compound over rulesets}
\date\today

\author[1]{Eric Duch\^ene}
\author[1]{Marc Heinrich}
\author[2]{Urban Larsson}
\author[1]{Aline Parreau}

\affil[1]{Université Lyon 1, LIRIS, UMR5205, France\thanks{Supported by the ANR-14-CE25-0006 project of the French National Research Agency and the CNRS PICS-07315 project.}}
\affil[2]{The Faculty of Industrial Engineering and Management,
Technion - Israel Institute of Technology,
Israel}

\newif\ifNB
\NBfalse

\begin{document}
%
%
%
%

\maketitle

\begin{abstract}
	We study operators that combine combinatorial games. This field was initiated by Sprague-Grundy (1930s), Milnor (1950s) and Berlekamp-Conway-Guy (1970-80s) via  the now classical disjunctive sum operator on (abstract) games. The new class consists in operators for \emph{rulesets}, dubbed the \emph{switch-operators}. The  ordered pair of rulesets $(\R_1, \R_2)$ is \emph{compatible} if, given any position in $\R_1$, there is a description of how to move in $\R_2$. Given compatible $(\R_1, \R_2)$, we build the \emph{push-the-button} game $\R_1 \swComp \R_2$, where players start by playing according to the rules $\R_1$, but at some point during play, one of the players must \emph{switch} the rules to $\R_2$, by pushing the button `$\swComp$'. Thus, the game ends according to the terminal condition of ruleset $\R_2$. We study the pairwise combinations of the classical  rulesets {\sc Nim}, {\sc Wythoff} and {\sc Euclid}. In addition, we prove that standard periodicity results for {\sc Subtraction games} transfer to this setting, and we give partial results for a variation of {\sc Domineering}, where $\R_1$ is the game where the players put the domino tiles horizontally and $\R_2$ the game where they play vertically (thus generalizing the octal game 0.07).
	
\end{abstract}

\section{Gallimaufry--new combinations of games}


Combinatorial Game Theory (CGT) concerns combinations of individual games. The most famous example is the {\em disjunctive game sum} operator, $G+H$, which, given two games $G$ and $H$, consists in playing in either $G$ or $H$ until both games are exhausted. Several variations of this compound were defined by Conway in his reference book~\cite{onag} and Berlekamp Conway Guy in \cite{winningWays}. One such example is the {\em ordinal sum} operator, which behaves like a disjunctive sum, with the additional constraint that any move in $G$ annihilates the game $H$.

More recently, the so-called {\em sequential compound} operator was introduced~\cite{seqComp}. It consists in playing successively the ordered pair of games $(G, H)$, where $H$ starts only when $G$ is exhausted. These constructions are in this sense static: the starting position of the second game is fixed until the first game is exhausted. In this work, we study dynamic compounds, where the starting position of the second game depends on the moves that were made in the first game. As a consequence, our construction requires the full rulesets of the two games to build the compound.


In order to define such compounds, we will say that an ordered pair of combinatorial game rulesets $(\R_1, \R_2)$ is \emph{compatible} if, given any position of $\R_1$, there is a description of how to move in $\R_2$. Note that we do not give a formal definition of this notion deliberately, as in the current context, we will only consider the case where $\R_1$ and $\R_2$ have the same set of positions. Typically, a pair of rulesets $(\R_1, \R_2)$ is compatible whenever there is a function mapping any position of $\R_1$ into a position of $\R_2$. In our case, this function is just the identity.
Given a compatible pair of rulesets $(\R_1, \R_2)$ and a specially designed \emph{switch procedure} `$\Rightarrow$' (it can be a game in itself or anything else that declares a shift of rules), players start the game $\R_1 \Rightarrow \R_2$ with the rules $\mathcal{R}_1$. At their turn, a player can choose, instead of playing according to $\mathcal{R}_1$, to play in the switch procedure, and when this procedure has terminated, the rules are switched to $\mathcal{R}_2$ (using the current game configuration). Thus, we call this class of operators {\em switch  operators} (or switch procedures).


In the current paper, the switch procedure is called the {\em push-the-button} operator (for short \emph{push operator}). It consists in a button that must be pushed once and only once, by either player. Pushing the button counts as a move, hence after a player pushes the button, his opponent makes the next move. The button can be pushed any time, even before playing any move in $\R_1$, or when there is no move left in either game. In particular, if there is no move left in $\R_1$ and the button has not been pushed, then it has to be pushed before playing in $\R_2$ (the only way to invoke $\R_2$ is via the push-the-button move). In all cases, the game always ends according to the rules of~$\mathcal{R}_2$.

We first recall some definitions and results about combinatorial game theory needed in this paper, and then give in Section~\ref{sec:2} a formal definition of the push operator.
In Section~\ref{sec:miniline}, we motivate this operator by the resolution of a particular game called \Rthl, which we demonstrate is a push compound of the classical games of two-heap {\sc Nim} \cite{nim} and {\sc Euclid} \cite{euclid}. Moreover we show that the second player's winning positions are similar to those of another classical game, {\sc Wythoff Nim} (a.k.a {\sc Wythoff}) \cite{wythoff}. Then, in Section~\ref{sec:heapGames} we continue by studying various push compounds of these three classical games. We also prove that standard periodicity results for
{\sc Subtraction games} transfer to this setting.
We finish off by studying in Section~\ref{sec:cram} a variation of {\sc Domineering}, where $\R_1$ is the game where the players put the domino tiles horizontally and $\R_2$ the game where they play vertically (thus generalizing the octal game 0.07).

\subsection{More overview}
Examples of popular traditional rulesets which can be considered as sequential compounds of several rulesets include {\sc Three Men's Morris}, {\sc Picaria} \cite{picaria} (first place then slide to neighbours) and {\sc Nine men's morris} (which is often played in three phases, first place, then slide and perhaps capture, and at last one of the players can move their pieces freely).

In addition, there has recently been scientific progress in building new games from combinations of the rules themselves. A nice illustration of this process is the game {\sc Clobbineering} \cite{Kyle}, whose rules are defined as a kind of union of the rules of the games {\sc Clobber} and {\sc Domineering}~\cite{Lessons}. Indeed, in {\sc Clobbineering}, the players choose either a {\sc Clobber} or a {\sc Domineering} move. Both types of moves are made on the same board.

One can also mention the more recent {\sc Building Nim}~\cite{building}, which behaves like {\sc Nine men's morris} (i.e., the starting position of $\R_2$ is a final position of $\R_1$). However, if all these games are fun to play with, and may sometimes be solved thanks to ad hoc techniques, there is no known work relating to a formal definition of a compound of compatible rules.

A similar construction called {\em conjoined ruleset} has been recently introduced in~\cite{RiMe} and used to create the games {\sc Go-Cut} and {\sc Sno-Go} which are compounds of the rulesets {\sc NoGo, Cut-Throat} and {\sc Snort}. In a conjoined ruleset, the players start playing according to a first ruleset until they reach a terminal position. Then, the game continues from the current position using the second ruleset. This case, as well as {\sc Building Nim}, can be viewed as examples of a switch operator where the switch procedure automatically changes the rules to the second ruleset when the players reach a terminal position.

Like in a conjoined ruleset, the push-the-button operator can be considered as a {\em sequential compound of two sets of rules}, in the sense that the two compound sets are considered successively while playing. The main difference between these two switch procedures is that in the push-the-button operator, the rules can be changed at any time during the game. Moreover, one great interest of the push-the-button operator is its correlation with games allowing pass moves. Indeed, such games can be considered as a special instance of this operator, in analogy with how classical sequential compounds generalize the mis\`ere convention in combinatorial game theory~\cite{seqComp}.


\subsection{Impartial games and their outcomes}
The fundament of a combinatorial game is the set of game positions; a `game board' together with some pieces (to place, move or remove etc). The description of how to manipulate the pieces on the game board is the ruleset.\footnote{Note that Definition~\ref{def:ruleset} deals with impartial rulesets, i.e., rulesets where both players always have the same available moves during the play. However, it can be extended to partizan (i.e., non-impartial) games by separating the ruleset into two parts, each corresponding to the moves available to one player.} 
\begin{definition}[Ruleset]\label{def:ruleset}
Given a set $\Omega$ of game positions, an impartial \defemph{ruleset} $\R$ over $\Omega$ is a function ${\R : \Omega \rightarrow 2^\Omega}$.
\end{definition}
The function $\R$ gives the set of move options from each position in $\Omega$. This ruleset can be extended to a function over sets of positions by saying that, for a set of positions $A \subset \Omega$, $\R(A) = \bigcup_{g \in A} \R(g)$.

Specifically, given a position $g\in \Omega$, the elements of $\R(g)\subset \Omega$ are called the {\em options} of $g$. Our study is carried out in the usual context of {\em short} games, i.e., games with a finite number of positions, and where no position is repeated during the play. The \emph{followers} of $g  \in\Omega $ are the positions in sets of the form $\R^i(g)$, where $\R^i = \R(\R^{i-1})$, for $i > 0$,  and $\R^0$ is the identity function. We say that a ruleset $\R$ is \emph{short}, if any game $g\in \Omega$ is short. That is, for any $g \in \Omega$, there is an $i$ such that $\R^i(g) = \emptyset $.

According to \cite{gameTheory} (Chapter I), a {\em game} is defined as the combination of a ruleset $\R$ over $\Omega$ with a specific starting position $g\in \Omega$ (but without an explicit definition of `ruleset'). The game will be denoted by $(g)_\R$ (or simply $g$ if the context is clear). 

\begin{example}
Consider the classical game of {\sc Nim}, which is played with heaps of finitely many tokens. At their turn, players can remove any number of tokens from one single heap (at least one and at most a whole heap). The game ends when there is no token left. Any position, with say $k$ heaps, can be represented by a $k$-tuple of non negative integers, $(n_1, \ldots, n_k)$, with the set of options
$$ \bigcup_{1 \leq j \leq k} \{(n_1, \ldots n_{j-1}, i, n_{j+1}, \ldots, n_k)\mid \ 0 \leq i < n_j \}.$$
It is easy to check that $\R=$ {\sc Nim} is a short ruleset.
\end{example}

In what follows, rulesets will be written with small capitals, or using the letters $\R, \R_0, \R_1, \ldots$ while $g, g', h$ are game positions. For a given ruleset, there are usually two main conventions to determine the winner:
\begin{itemize}
\item \defemph{normal} play: the first player with no move available loses the game.
\item \defemph{misere} play: the first player with no moves available wins the game (i.e., the last to play loses the game).
\end{itemize}

A short impartial game has two possible outcomes: either the first player has a strategy to ensure a win, or the second player has one. In normal (misere) play, we denote $o_\R(g)$ ($o^-_\R(g)$) the outcome of the game $(g)_\R$, with $\N$ if the first player has a winning strategy, and $\P$ otherwise (for \textbf Next player and \textbf Previous player respectively). In this paper, unless specified otherwise, we will consider games under the normal convention. A well-known property of impartial games yields necessary and sufficient conditions to characterize the outcomes of a game.
\begin{proposition}\label{prop:outcome}
  Given a game $(g)_\R$ we have:
  $$
  o_\R(g) = \left \{ \begin{array}{cl}
  \P& \mbox{if}\;\; \forall g'\in \R(g),\ o_\R(g') = \N  \\
  \N& \mbox{if}\;\; \exists g'\in \R(g),\  o_\R(g') = \P
  \end{array}\right.
  $$
  $$
  o^-_\R(g) = \left \{ \begin{array}{cl}
  \P& \mbox{if}\;\; \R(g) \neq \emptyset \ \mbox{ and } \ \forall g'\in \R(g),\ o^-_\R(g') = \N  \\
  \N& \mbox{if}\;\; \R(g) = \emptyset \  \mbox{ or } \ \exists g'\in \R(g),\ o^-_\R(g') = \P.
  \end{array}\right. 
  $$
\end{proposition}
We will use the following notation to define the \P-positions of a given ruleset:
\begin{definition}[\P-positions]\label{def:ppos}
Let $\R$ be a ruleset over $\Omega$. The set of \P-positions for $\R$, in normal and misere play, are the sets
$$ P = \{g \in \Omega \mid o_{\R}(g) = \P\}$$
and
$$ P^- = \{g \in \Omega \mid o^-_{\R}(g) = \P\}$$
respectively.
\end{definition}

\subsection{Disjunctive sums and Sprague-Grundy theory}
The disjunctive sum of two games $(g_1)_{\R_1}$ and $(g_2)_{\R_2}$ consists in putting the two games side by side, and at each turn a player makes a move either in $g_1$ or in $g_2$. In our approach, since we focus on rulesets rather than games, we slightly reformulate the standard definition (content remains the same).

\begin{definition}[Disjunctive sum]
	\label{def:sum}
	Consider two rulesets $\R_1$ over $\Omega_1$ and $\R_2$ over $\Omega_2$. Their \emph{disjunctive sum}, denoted by $\R_1 + \R_2$ is the ruleset defined over $\Omega_1 \times \Omega_2$ such that, for all $(g_1, g_2)\in \Omega_1\times \Omega_2$,
	$$ (\R_1 + \R_2)(g_1,g_2) = \{(g'_1,g_2), \ g'_1 \in \R_1(g_1)\} \cup \{(g_1,g'_2), \ g'_2 \in \R_2(g_2) \}.$$
	Analogously, the disjunctive sum of a pair of given games, over given rulesets, is denoted $(g_1)_{\R_1} + (g_2)_{\R_2}$.
\end{definition}
This notion of disjunctive sum arises quite naturally for several rulesets. For example, the game of {\sc Nim} played on a position with $k$ heaps can be seen as a disjunctive sum of $k$ different $1$-heap {\sc Nim} rulesets. From this construction, a very natural question arises: what can we say on the outcome of $(g_1)_{\R_1} + (g_2)_{\R_2}$ by studying the two games separately? An answer to this question is obtained by refining the notion of outcome by assigning to a game a numeric value called the {\em Sprague-Grundy value}~\cite{grundy, sprague} (for short {\em value}):

\begin{definition}[Value]
Given a game $(g)_\R$, its value $\G_\R(g)$ is a non negative integer (coded in binary) defined by: $$\G_\R(g) = \mex\{\G_\R(g'), g' \in \R(g)\}$$ where $\mex(S)$ (for Minimum EXcluded value) is the smallest non negative integer not in the set $S$.
\end{definition}

\begin{proposition}[\cite{grundy, sprague}]
	\label{prop:grundy}
	Consider games $(g)_\R$, $(g_1)_{\R_1}$ and $(g_2)_{\R_2}$. Then
	\begin{itemize}
		\item the game $(g)_\R$ has outcome $\P$ if and only if its value is $0$.
		\item the value of $(g_1)_{\R_1} + (g_2)_{\R_2}$ is $\G_{\R_1}(g_1) \oplus \G_{\R_2}(g_2)$ where $\oplus$ is the bitwise XOR operator.
	\end{itemize}
\end{proposition}

Values are a refinement of the notion of outcome, since the \P-positions of a ruleset are exactly the games with value $0$, and the outcome of a sum of games can be easily computed from the values of each game separately.

When the ruleset is clear from the context, we might omit the subscript in our notations, and write simply $o(g)$ and $\G(g)$ to denote respectively the outcome and the value of the game $(g)_\R$.

\section{The switch compound: push-the-button}\label{sec:2}
Given two rulesets $\R_1$ and $\R_2$ over the same set of positions $\Omega$, clearly both orders of the pair of rulesets are compatible. Unless otherwise specified, players start the game according to $\R_1$. At their turn, a player can decide, instead of playing according to $\R_1$, to push a single allocated button, which switches the rules to the second ruleset $\R_2$. The button must be pushed exactly once, at some point during play, by either player. It can be pushed before any move has been made, at some intermediate stage, or even if there is no other move available. We call this operator, the \emph{push operator}, denoted by `$\swComp$'. By pushing the button, the current position, which is a position of the first ruleset, becomes the starting position for the second ruleset.

\begin{definition}[The push operator]\label{def:PB}
Let $\R_1$ and $\R_2$ be two rulesets over the same set of positions~$\Omega$. The push-the-button ruleset (for short \emph{push ruleset}) $\R_1 \swComp \R_2$ is defined over $\{ 1,2\} \times \Omega$ by:

$$ (\R_1 \swComp \R_2) :
\begin{array}{lcc}
(1,g) & \longmapsto & \{(1,g'), g' \in \R_1(g) \} \cup (2,g) \\
(2,g) & \longmapsto & \{ (2,g'), g' \in \R_2(g)\}.
\end{array} $$
\end{definition}

Once the button is pressed, the game is played according to the rules $\R_2$. In other words, we have $(2,g)_{\R_1 \swComp \R_2} = (g)_{\R_2}$. Consequently, when the ruleset is of the form $\R_1 \swComp \R_2$, the interesting positions are those of the form $(1,g)$ for $g \in \Omega$. Thus, to simplify notation, we write $(g)_{\R_1 \swComp \R_2}$ instead of $(1,g)_{\R_1 \swComp \R_2}$.

\subsection{General properties of the push compound}
\label{subsec:pbprops}

We now state a few general properties of the push operator. The first one concerns the equivalence between the push-compound of a ruleset with itself and games including a (compulsory) pass move, that is a {\sc Nim} heap of size 1.

\begin{lemma} Let $\R$ be a ruleset over a set of positions $\Omega$. We have:
$$ \forall g \in \Omega, \ (g)_{\R \swComp \R} = (g)_{\R} + \ast,$$
where $\ast$ is a single {\sc Nim} heap of size $1$.
\end{lemma}
\begin{proof}
We combine Definitions~\ref{def:PB} and \ref{def:sum}. Since pushing the button does not change the rules, it mimics $\ast$ in disjunctive sum with any game $g\in \Omega$.
\end{proof}
The following characterization of the \P-positions of push rulesets is the natural extension of the normal play part of Proposition~\ref{prop:outcome}:
\begin{proposition}
\label{prop:PposPush}
Given a push ruleset $\R_1 \swComp \R_2$ over a set $\Omega$, and a position $g \in \Omega$, we have:
$$ o_{\R_1 \swComp \R_2}(g) = \P \iff o_{\R_2}(g) = \N \text{ and }\ \forall g' \in \R_1(g), \ o_{\R_1 \swComp \R_2}(g')= \N .$$
\end{proposition}
\begin{proof}
The moves from $(g)_{\R_1 \swComp \R_2}$ are either to $(g)_{\R_2}$ if the button is pushed, or to $(g')_{\R_1 \swComp \R_2}$ if the move is played according to $\R_1$.
\end{proof}
We have the following reformulation of Proposition \ref{prop:PposPush}. It will be used further to characterize the \P-positions of several push rulesets.

\begin{corollary}
\label{cor:PposPush}
Let $\R_1 \swComp \R_2$ be a push ruleset over $\Omega$. A subset $P\subset\Omega$ is the set of \P-positions of $\R_1 \swComp \R_2$ if and only if:
\begin{enumerate}[label=(\roman{enumi})]
\item $\forall g \in P, \ o_{\R_2}(g) = \N,$
\item $\forall g \in P$, $\R_1(g)\subset \Omega\setminus P$,
\item $\forall g \in \Omega \setminus P$, either $o_{\R_2}(g) = \P$ or $\R_1(g)\cap P\neq \emptyset$.
\end{enumerate}
\end{corollary}
\begin{proof}
Obvious, according to Definition~\ref{def:ppos} and Proposition~\ref{prop:outcome}.
\end{proof}
In some cases, the \P-positions of $\R_1 \swComp \R_2$ can be determined directly from the \P-positions of $\R_1$. Indeed, adding the possibility to change the rules can be seen as a way to disrupt the game $\R_1$. The following results give sufficient conditions for which such a disruption preserves the \P-positions of $\R_1$.

\begin{proposition}
\label{prop:push}
Let $\R_1$ and $\R_2$ be  two rulesets over $\Omega$ and let $g$ in $\Omega$ be a game position. We denote by $P_1$, $P_2$ and $P_1^-$ the set of $\P$-positions for respectively $\R_1$ and $\R_2$ under normal convention, and $\R_1$ under misere convention.

The ruleset $\R_1 \swComp \R_2$ satisfies the following properties:

\begin{enumerate}[label=(\roman{enumi})]
\item If $o_{\R_2}(g) = \P$, then $o_{\R_1 \swComp \R_2}(g) = \N$.
\item If $P_1 \cap P_2 = \emptyset$, then $o_{\R_1 \swComp \R_2}(g) = \P \iff g \in P_1$.
\item If $P_1^- \cap P_2 = \emptyset$ and  $\R_1(g) = \emptyset$ implies $g \in P_2$, then $o_{\R_1 \swComp \R_2}(g) = \P \iff g \in P_1^-$.
\end{enumerate}
\end{proposition}

\begin{proof}
\begin{enumerate}[label=(\roman{enumi})]
\item Pushing the button is a winning move.
\item We check that $P_1$ satisfies the three conditions $(i), (ii)$ and $(iii)$ of Corollary~\ref{cor:PposPush}. The condition $(i)$: for all $ g \in P_1, \ o_{\R_2}(g) = \N $, corresponds to the assumption we made that $P_1 \cap P_2 = \emptyset$. Since $P_1$ is the set of \P-positions for the ruleset $\R_1$, there is no move, according to $\R_1$, from a position in $P_1$ to another one, hence $(ii)$ holds. Moreover for any position not in $P_1$, there is a move, according to $\R_1$, to a position in $P_1$ and $(iii)$ holds.
 
\item The argument is essentially the same as above. If a position $g$ has no move for $\R_1$, then by hypothesis, pushing the button is a winning move. Otherwise, $g$ has at least one option. If $g \in P_1^-$, then all moves, according to $\R_1$, are out of $P_1^-$, and by hypothesis, pushing the button is a losing move. If $g \not \in P_1^-$, then there exists at least one move, according to $\R_1$, landing in $g' \in P_1^-$.
\end{enumerate}
\end{proof}

This result gives sufficient conditions for which the \P-positions of $\R_1$ remain unchanged by adding the possibility to switch the rules to $\R_2$. 
In the following sections, we study some games which use this property.

\section{\Rthl}
\label{sec:miniline}

We now present the ruleset that is the source of the switch compound. It is called \Rthl\ (a.k.a. \textsc{Ruthlein}), and is is a variant of the ruleset {\sc Susen}~\cite{selfref}. Both are part of a larger family of rulesets called in the literature {\em non-invariant}~\cite{invariant} (or simply `variant') with the interesting subclass of {\em self-referential}~\cite{selfref} games. Previous naming conventions for classes of `variant' games are for example {\em pilesize dynamic} in~\cite{holshouser03, holshouser04}, and \emph{time and size dependent take-away game} in~\cite{flanigan82}.


\begin{definition}[\Rthl]
The game \Rthl\ is played on several heaps of tokens. At each turn, a player can remove from any heap a number of tokens which is a positive multiple of the smallest non-zero heap. The heap sizes must remain non-negative.
\end{definition}

\Rthl\ on two heaps is very similar to the well-known game \textsc{Euclid}, with the only difference that the latter game requires positive heap sizes.

\begin{definition}[\textsc{Euclid}]
\label{euclidDef}
The game \textsc{Euclid} is played on two heaps of tokens. A player can remove a positive multiple of the smallest heap from the largest heap. The heap sizes must remain positive.
\end{definition}

Played on two heaps, \Rthl\ has only two differences with \textsc{Euclid}:
\begin{itemize}
\item In \Rthl, the game ends when both heaps are zero, whereas in \textsc{Euclid} it ends when they are equal.
\item In \Rthl, it is possible to remove the smallest heap in a single move.
\end{itemize}

One can easily check that these two differences do not modify the $\P$-positions. Indeed, in the $2$-heap case, removing the smallest heap is always a losing move in \Rthl. Hence the $\P$-positions of $2$-heap \Rthl\ with non-zero coordinates, and the $\P$-positions of \textsc{Euclid} are the same. They are given by the following result:

\begin{proposition}[\cite{euclid}]
\label{euclidPos}
A position $(a,b)$, with $1\leq a \leq b$ is a \P-position of \textsc{Euclid} if and only if $\frac{b}{a}<\Phi$, where $\Phi$ is the golden ratio.
\end{proposition}

Played on three heaps, our game is similar to the game \textsc{$3$-Euclid} introduced in \cite{3euclid}, and to some of its variations studied in \cite{variations3Euclid}. The main difference with the game \Rthl\  is that \textsc{$3$-Euclid} requires the heap sizes to remain positive. In particular, the number of heaps does not vary during the game. In our version, it is allowed to completely remove the smallest heap, and thus to decrease the total number of heaps. An example is given in Figure~\ref{fig:rthlex}. 

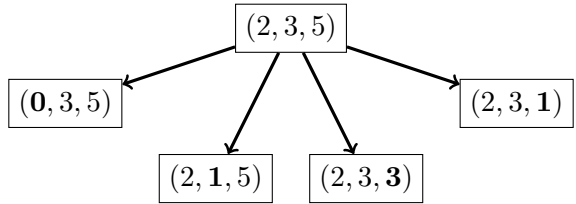
\begin{figure}[!ht]
	\centering
	\begin{tikzpicture}
		\node[draw] (g) at (0,0) {$(2,3,5)$} ;
		\node[draw] (g1) at (-3,-1) {$(\mathbf{0},3,5)$} ;
		\node[draw] (g2) at (-1,-2) {$(2,\mathbf{1},5)$} ;
		\node[draw] (g3) at (1,-2) {$(2,3,\mathbf{3})$} ;
		\node[draw] (g4) at (3,-1) {$(2,3,\mathbf{1})$} ;
		\foreach \i in {1,2,3,4} {
			\draw[->, very thick] (g) -- (g\i) ;
		}
	\end{tikzpicture}
	\caption{\label{fig:rthlex} Example of the possible moves for \Rthl\ on position $(2,3,5).$}
\end{figure}

\subsection{$\Rthl(1,a,b)$}
We call $\Rthl(1,a,b)$ the game \Rthl\ played on three heaps of tokens where the smallest heap has size $1$. The link between this game and the push-the-button operator is given in the following proposition:

\begin{proposition}
The game \Rthl\ played on position $(1,a,b)$ has the same outcome as the push ruleset $\textsc{Nim}  \swComp \textsc{Euclid}$ on position $(a,b)$.
\end{proposition}

	First, remark that \textsc{Euclid} and \textsc{Nim} are not played exactly on the same set of positions: the position $(0,i)$ is not a valid position for \textsc{Euclid}, while it is a valid position for \textsc{Nim}. Consequently, what we call \textsc{Euclid} here, and in all the rest of Section~\ref{sec:miniline}, is the variant where there is a single move from $(0,i)$ to $(0,0)$ if $i > 0$. In particular, the position $(0,i)$ is an \N-position.

\begin{proof}
	Before the heap with $1$ token is taken, it is possible to remove any number of tokens from any of the two other heaps, hence the rules are those of \textsc{Nim} on two heaps. Once the heap with size one is removed, the rules are essentially those of \textsc{Euclid}, up to some minor differences, which, as mentioned before, do not modify the set of $\P$-positions. Removing the heap of size one corresponds to pushing the button.
\end{proof}

The rest of this section will be dedicated to finding characterizations of the \P-positions of $\textsc{Nim}  \swComp \textsc{Euclid}$. By the previous proposition, this also gives a characterization of the \P-positions for $\Rthl(1,a,b)$.

Surprisingly, the \P-positions of $\textsc{Nim}  \swComp \textsc{Euclid}$ are very similar to the \P-positions of another well-known game:

\begin{definition}[\textsc{Wythoff}'s game]
\label{wythoffDef}
The game \textsc{Wythoff} is played on two heaps of tokens. A move consists in either removing a positive number of tokens from one heap, or removing the same number of tokens from both heaps.
\end{definition}

In the following, let $a_n = \lfloor\Phi n \rfloor$ and $b_n = \lfloor\Phi n \rfloor+n$.

\begin{proposition}[\cite{wythoff}]
\label{wythoffPpos}
The \P-positions of \textsc{Wythoff} are: $$\{(a_n, b_n) ,(b_n,a_n),\ n \geq 0\}.$$
\end{proposition}
According to the proposition above, the pairs $(a_n, b_n)$ for $n \geq 0$ are generally called \textsc{Wythoff} pairs.

\begin{definition}
Denote by $(F_n)_{n \geq 0}$ the Fibonacci sequence starting with ${F_0=0}$ and $F_1 = 1$. The sequence $(u_n)_{n \geq 0}$ is defined by $u_n = F_{n+1} -1$.

\end{definition}

\begin{proposition}
	\label{prop:wythoffPairs}
	For all $n \geq 0$, the pair $(u_{2n}, u_{2n+1})$ is a \textsc{Wythoff} pair.
\end{proposition}
This is a well known result. Since we will reuse the proof afterwards, we put the details of it below.
\begin{proof}
	Let $\Phi$ denote the golden ratio. Using the closed form formula for the Fibonacci numbers, we can write that:
	\begin{align*}
		F_{n+1} - \Phi F_n & = \frac{\Phi^{n+1} - (- \Phi)^{-n-1}}{2 \Phi -1} - \Phi \frac{\Phi^{n} - (- \Phi)^{-n}}{2 \Phi -1} \\
		& = \frac{1}{(-\Phi)^n}\frac{(\Phi - \frac 1 \Phi)}{2 \Phi - 1} = \frac{1}{(-\Phi)^n(2 \Phi -1)}.
	\end{align*}
	This gives $F_{2n +1} - \frac{1}{\Phi^{2n} (2\Phi -1)} = \Phi F_{2n}$, and by taking the integer part on both sides of the equality, we obtain $u_{2n} = \lfloor \Phi F_{2n}\rfloor$. Additionally we have $$u_{2n +1} = F_{2n +2} - 1 = F_{2n+1} - 1 + F_{2n} = \lfloor \Phi F_{2n}\rfloor + F_{2n} = \lfloor \Phi^2 F_{2n}\rfloor. $$
	Consequently, the equality $(u_{2n},u_{2n+1}) = (\lfloor \Phi F_{2n}\rfloor,\lfloor \Phi^2 F_{2n}\rfloor)$ proves that it is a \textsc{Wythoff} pair.
\end{proof}

\begin{theorem}
\label{rthl1}
The set of \P-positions of $\textsc{Nim} \swComp \textsc{Euclid}$ is given by:
$$P =  \{(a_n,b_n), n \geq 0\} \setminus \{(u_{2n}, u_{2n+1}), \ n \geq 0\} \bigcup \{ (u_{2n+1}, u_{2n+2}), \ n\geq 0 \}.$$
\end{theorem}
This set coincides with the \P-positions of \textsc{Wythoff}, except for a very small fraction of the positions. The first \P-positions of the two games are given in Table~\ref{tab:ppos}.

\begin{proof}
	Denote by $(a'_n, b'_n)$ the positions of $P$, with $b'_n \geq a'_n$, reordered by increasing $a'_n$. Denote by $A$ and $B$ the two sets defined by $A = \{ a'_n,\  n \geq 0 \}$ and $B = \{ b'_n, \ n \geq 0\}$.
	We know by \cite{wythoff} that the sets $\{ a_n,\  n \geq 0 \}$ and $\{ b_n, \ n \geq 0\}$ are complementary. Additionally, using  Proposition~\ref{prop:wythoffPairs} and the fact that the $u_i$ for $i \geq 1$ are all distinct, we can deduce that $A$ and $B$ are complementary. This implies that there is no move, according to \textsc{Nim} from a position of $P$ to an other position of $P$.
	
	In order to apply Corollary~\ref{cor:PposPush}, we only need to show that: $(i)$ pushing the button on a position $(a'_n, b'_n)$ is a losing move, and $(ii)$ from any position $(a,b) \not \in P$, there is either a move, according to \textsc{Nim}, to a position in $P$, or pushing the button is a winning move.
	
	We start by showing that, for $n \geq 1$, the position $(a'_n, b'_n)$ satisfies the equality $b'_n = \lceil \Phi a'_n \rceil$. If $(a'_n,b'_n)$ is a \textsc{Wythoff} pair, this relation is a well known result (this is proved for example in Lemma~5 from \cite{raleigh} or in \cite{silber2}). Consequently, we only need to prove it in the case $(a'_n, b'_n) = (u_{2n+1}, u_{2n+2})$. Reusing the computations from the proof of Proposition~\ref{prop:wythoffPairs}, we know that:
	
	\begin{align*}
		u_{2n+2} - \Phi u_{2n +1} & = F_{2n + 3} - 1 - \Phi F_{2n + 2} + \Phi \\
		 & = \frac{1}{\Phi ^{2n +2}(2 \Phi -1)} - 1 + \Phi.
	\end{align*}
	This shows that  ${u_{2n+2} - \Phi u_{2n +1} \leq \Phi - 1 + \frac{1}{\Phi^2(2\Phi -1)} < 1}$ and additionally, ${u_{2n+2} - \Phi u_{2n +1} > 0}$, which proves that $u_{2n +2} = \lceil \Phi u_{2n+1} \rceil$.
	
	We have $(a'_0, b'_0) = (0,1)$, and pushing the button from this position is a losing move. If $n \geq 1$, then $\frac{b'_n}{a'_n} > \Phi$, consequently, by Proposition~\ref{euclidPos}, pushing the button is also a losing move.

	Now suppose that we have $(a,b) \not \in P$. We want to show that either there is a move, according to \textsc{Nim}, to a position in $P$, or pushing the button is a winning move. If $a =0$, then either $b=0$ and pushing the button is a winning move, or $b > 0$, and there is a move to $(0,1) \in P$. Consequently, we can assume $a > 0$.
	Suppose $a = b'_n$ for some $n$. Thus we have $b \geq a = b'_n > a'_n$, and there is a move, according to \textsc{Nim}, to the position $(b'_n, a'_n)$. Otherwise, since $A$ and $B$ are complementary, we have $a= a'_n$ for some $n$. If $b > b'_n$, then again, there exists a move, according to \textsc{Nim}, to $(a'_n, b'_n)$. Otherwise, we have  $b < b'_n$. Since $b'_n = \lceil \Phi a'_n \rceil$, we must have $\frac b a < \Phi$, which implies by Proposition~\ref{euclidPos} that pushing the button is a winning move.
\end{proof}

\begin{table}[!ht]
	\begin{adjustbox}{width=\textwidth,center}
		\setlength\tabcolsep{1.5pt}
		\def\arraystretch{1.2}
		\begin{tabular}{|l|ccccccccccccccc|}
			\hline
			\textsc{Wythoff} 						& (0,0) &       & (1,2) &       & (3,5) & (4,7) & (6,10) &        & (8,13)  & (9,15) & (11,18) & (12,20) & (14,23) & (16,26) & (17,28)  \\
			$\textsc{Nim} \swComp \textsc{Euclid}$ 	&       & (0,1) &       & (2,4) & (3,5) &       & (6,10) & (7,12) & (8,13)  & (9,15) & (11,18) &         & (14,23) & (16,26) & (17,28)   \\
			\hline
			
		\end{tabular}
	\end{adjustbox}
	\caption{\label{tab:ppos} Sequence of the first \P-positions for the games \textsc{Wythoff} and $\textsc{Nim} \swComp \textsc{Euclid}$. Some blanks were inserted to highlight the similarities between the two sequences. }
\end{table}

There exists several ways to characterize the \P-positions for the game \textsc{Wythoff}. Quite naturally, similar characterization also exist for $\textsc{Nim} \swComp \textsc{Euclid}$. They are given in the following theorems. The proofs are very similar to those for \textsc{Wythoff} and are left in the Appendix.

\begin{restatable}{theorem}{recCaract}
\label{th:rthlrec}
The \P-positions of $\textsc{Nim} \swComp \textsc{Euclid}$ are the $(A_n,B_n)$ where $A_n$ and $B_n$ are given by the following recurrence equations:
$$ \left\{ \begin{array}{l}
A_0 = 0, \quad B_0 = 1, \\
A_n = \mex \{ A_i, B_i, i < n\},  \\
B_n = \lceil \Phi A_n \rceil.
\end{array}  \right.$$
\end{restatable}

As is the case for {\sc Wythoff}, we finally give a third characterization of the \P-positions using the following Fibonacci numeration system:
\begin{definition}
Any natural number $n$ can be written uniquely as a sum of Fibonacci numbers:

$$ n = \sum_{i \geq 2} x_i F_i  $$
with $x_i \in \{0,1\}$, without two consecutive $1$s (i.e. $x_i = 1 \Rightarrow x_{i+1} = 0$).
\end{definition}
For a given integer $n$, its Fibonacci representation is written as the bit string $x_k x_{k-1} \ldots x_2$. We write $(10)^*$ for bit strings of the form $101010\ldots10$.

\begin{restatable}{theorem}{fiboCaract}
As in Theorem~\ref{th:rthlrec}, denote by $A = \{ A_n \}_{n \geq 0}$ and $B = \{ B_n\}_{n \geq 0}$ the sets of coordinates of the \P-positions of $\textsc{Nim} \swComp \textsc{Euclid}$. Given a non-negative integer $x$, we have $x \in A$ if and only if its Fibonacci representation~$s$:
\begin{itemize}
\item either ends with an even number of $0$s, and $s \neq (10)^*1$,
\item or $s = (10)^* $.
\end{itemize}

Moreover, if $s_A$ is the Fibonacci representation of some $A_n$, then the Fibonacci representation $s_B$ of the corresponding $B_n$ satisfies:
\begin{itemize}
\item $s_B = s_A0$ (i.e. a shift with $0$ of $s_A$) if $s_A \neq (10)^*$,
\item $s_B = s_A1$ (i.e. a shift with $1$ of $s_A$) if $s_A = (10)^*$.
\end{itemize}
\end{restatable}

\begin{remark}
  Although the \P-positions of \Rthl\ with the smallest heap of size $1$ are well characterized, the values seem to have a much more complicated structure (see Figure~\ref{fig:grundyRthl}). Note that the patterns we obtained seem very close to those obtained in \cite{colorfullFamilies} about variations of the game \textsc{Nim}.
\end{remark}

\begin{figure}[!h]
	\centering
	\ifNB
\else
	\includegraphics[width=.8\textwidth]{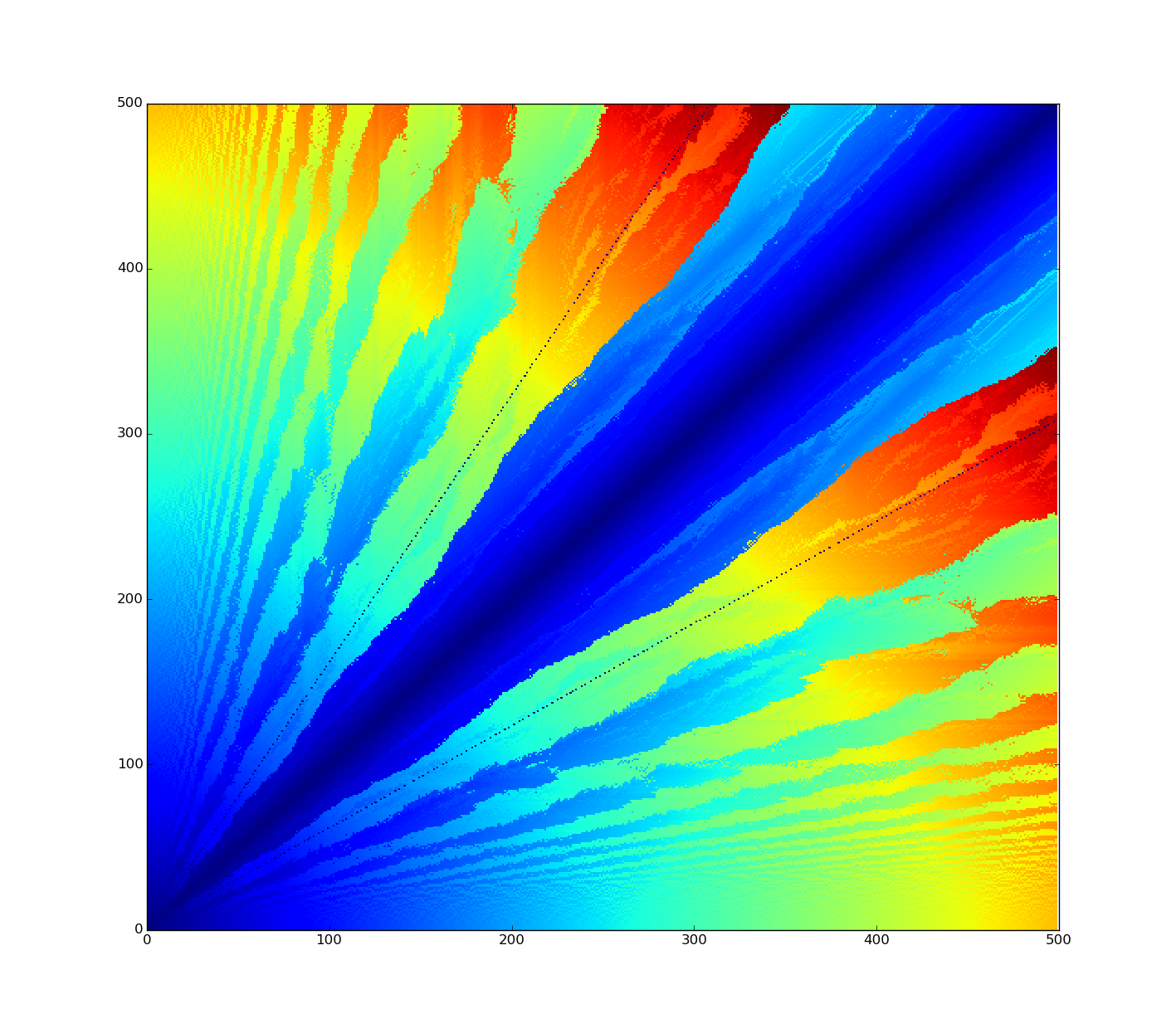}
	\caption{\label{fig:grundyRthl} Values for positions $(1,a,b)$ of $3$-heaps \Rthl. Color blue corresponds to small values while yellow and red correspond to high values. The \P-positions can be seen along the line $y = \Phi x$ and its symmetric. Other small positive $\G$-values seem to be close to the diagonal $y = x$.}
	\fi
\end{figure}

\subsection{3-heaps \Rthl: general case}

We have seen that $\Rthl(1,a,b)$ has the same outcome as the push ruleset $\textsc{Nim} \swComp \textsc{Euclid}$. In the more general case where the smallest heap contains $a>1$ tokens, the game cannot be viewed as a push ruleset. Indeed, there are three 'special moves' changing the rules of the game: removing the heap of size $a$, or decreasing one of the two other heaps to a size less than $a$. For this reason, the \P-positions seem more complicated to compute in the general $3$-heaps case. However we can still prove some properties on the \P-positions. These results are given in the following theorems. Since this part is less relevant to our study of push rulesets, only the statements of the theorems are given; the proofs are left in the Appendix.

\begin{restatable}{theorem}{ineqPpos}
	\label{th:ineqPpos}
Let $a, b$ and $c$ be three integers with $c\geq b \geq a > 0$ such that $(a,b,c)$ is a \P-position of \Rthl. The triple $(a,b,c)$ satisfies:
$$\lceil \Phi b \rceil \leq  c \leq \lceil \Phi b \rceil + a-1.$$
\end{restatable}

\begin{restatable}{theorem}{moduloPpos}
Given two positive integers $a$ and $b$ with $b \geq a$, there are exactly $a$ values $c$ such that $(a,b,c)$ is a \P-position of \Rthl, one for each congruence class modulo $a$.
\end{restatable}

Note that in the last theorem, we do not necessarily have $c \geq b$. In particular, not all positions satisfying the properties of Theorem~\ref{th:ineqPpos} are \P-positions.

\section{Push compounds of heap games}
\label{sec:heapGames}

In this section we investigate the push operator on several other well-known impartial rulesets played on heaps of tokens.

\subsection{Push compounds of {\sc Nim}, {\sc Wythoff} and {\sc Euclid}}

We start with examples of push rulesets played on two heaps of tokens. The set of positions of these games can be seen as the integer lattice $\mathbb N^ 2$. More precisely we will exhibit a full characterization of the \P-positions when we combine the rulesets of \textsc{Nim}, \textsc{Wythoff}, and \textsc{Euclid} (see Definitions \ref{euclidDef} and \ref{wythoffDef}).
Recall that the case $\textsc{Nim} \swComp \textsc{Euclid}$ was already considered in the previous section; it corresponds to $3$-heap \Rthl\ with the smallest heap of size $1$. The three results below cover the other compounds of \textsc{Nim} with \textsc{Euclid} and \textsc{Wythoff}.

\begin{proposition}
The \P-positions of $\textsc{Nim} \swComp \textsc{Wythoff}$ is the set $P_{N\swComp W}$ defined as:
$$ P_{N\swComp W} = \{(0,1); (1,0) ; (k,k), \ k \geq 2\}.$$
\end{proposition}

\begin{proof}
The set $P_{N\swComp W}$ is known as the set of \P-positions of \textsc{Nim} under misere convention \cite{nim}.
Using Proposition \ref{prop:push} (iii), one only needs to show that:
\begin{itemize}
\item If $g$ is a position with no option for \textsc{Nim}, then it is a \P-position for \textsc{Wythoff}. This is immediate because the only position with no option for \textsc{Nim} is $(0,0)$.
\item All the positions of $P_{N\swComp W}$ are winning for \textsc{Wythoff}. This also holds since the \P-positions of \textsc{Wythoff} are of the form $(\lfloor\Phi n\rfloor, \lfloor\Phi n\rfloor+n)$ (see Proposition \ref{wythoffPpos}), and none of these positions is in $ P_{N\swComp W}$.
\end{itemize}
\end{proof}

\begin{proposition}
	A position $(a,b)$ with $a\leq b$ is a \P-positions of $\textsc{Euclid} \swComp \textsc{Nim}$ if and only if one of the following assumptions holds:
\begin{itemize}
\item $\frac b a < \Phi$ and $\frac b a \neq \frac{F_{2i+2}}{F_{2i+1}}$ for any $i$,
\item $\frac b a = \frac{F_{2i+1}}{F_{2i}}$ for some $i$,
\end{itemize}
where the $F_i$ are the Fibonacci numbers.
\end{proposition}

\begin{proof}
If we note $P$ the set of positions defined by the union of the two above conditions, then $P$ is the set of \P-positions of \textsc{Euclid} in misere convention \cite{euclidMisere}.
To show that this also corresponds to the set of \P-positions of  $\textsc{Euclid} \swComp \textsc{Nim}$, using Property \ref{prop:push}, we only need to show that for any position in $P$, this position is winning for \textsc{Nim}, and that all the terminal positions of \textsc{Euclid} are \P-positions of \textsc{Nim}. This is straightforward since the \P-positions of \textsc{Nim} are of the form $(a,a)$, and $a/a = 1 = \frac{F_2}{F_1}$, thus $(a,a) \not \in P$. In addition, the terminal positions of \textsc{Euclid} are exactly these positions.
\end{proof}

\begin{proposition}
The \P-positions of $\textsc{Wythoff} \swComp \textsc{Nim}$ is the set $P_{W \swComp N}$ defined as follows:
$$P_{W \swComp N} = \{(a_n -1, b_n -1) \quad n \geq 1\}.$$ where $(a_n,b_n)$ are the \P-positions of \textsc{Wythoff}.
\end{proposition}
This is an example where the \P-positions of the compound ruleset differs from the \P-positions of the original ruleset (both in normal or misere convention). Therefore, Proposition~\ref{prop:push} cannot be applied.

\begin{proof}
According to Proposition~\ref{wythoffPpos}, the sets $A = \{ a_n - 1\}_{n\geq 1}$ and $B = \{ b_n - 1\}_{n\geq 1}$ are complementary sets and $(a_n - 1) - (b_n -1) = a_n - b_n = n$. As a consequence, there is no move, according to \textsc{Wythoff}, from one position in $\P_{W \swComp N}$ to another. Moreover, since $a_n - b_n = n > 0$ for $n \geq 1$, none of these elements are \P-positions of {\sc Nim}. Hence, to prove that this set is the set of \P-positions, it only remains to show that for any position not in $\P_{W \swComp N}$, either pushing the button is a winning move, or there is a move according to \textsc{Wythoff} ruleset to a position in $\P_{W \swComp N}$.
The proof is essentially the same as for the classical \textsc{Wythoff}.

Take a position $(a,b) \not \in \P_{W \swComp N}$ with $0\leq a \leq b$. If $a = b$, changing the rules to \textsc{Nim} (by pushing the button) is a winning move. Suppose now that $a < b$. Since $A$ and $B$ are complementary, there are two possibilities:
\begin{itemize}
	\item $a = b_n-1$ for some $n \geq 1$, then $b > a = b_n-1 \geq a_n-1$. In this case there is a move to $(b_n -1, a_n -1)$ by playing on the second heap.
	\item $a = a_n -1$ for some $n \geq 1$. If $b > b_n -1$, then there is a move to $(a_n -1, b_n -1)$ by playing on the second heap. Otherwise $b < b_n -1$, and let $m = b -a$. We have $ m < b_n -1 - a_n -1 = n$. Consequently, by choosing $k = a_m - a_n$, the move to $(a-k,b-k) = (a_m -1, b_m -1)$ is a winning move.
\end{itemize}
\end{proof}

\subsection{\textsc{Subtraction} games}

We now consider the push operator applied to \Subs\ rulesets.
\begin{definition}
	Let $S$ be a finite set of positive integers. The ruleset $\Subs(S)$ is played with one heap of tokens. At his turn a player can remove $v \in S$ tokens from the heap, provided there are at least $v$ tokens in the heap.
\end{definition}

For a ruleset $\R = \Subs(S)$ for some fixed set $S$, it is common to look at the sequence of values $(\G_\R(n))_{n \geq 0}$. It is already known that the sequence of values of \Subs\ games are periodic~\cite{winningWays}. A special case of this result is given in the proposition below:

\begin{proposition}[\cite{winningWays}]
	\label{prop:subsk}
	Let $k$ be a positive integer, and $\R = \Subs(\{1,\ldots k\})$ be a \Subs\ ruleset. The sequence of outcomes in normal and misere convention are periodic with period $k+1$. More precisely:
	\begin{itemize}
		\item In normal convention, $o_\R(n) = \P \iff n \equiv 0 \bmod (k+1)$.
		\item In misere convention, $o^-_\R(n) = \P \iff n \equiv 1 \bmod (k+1)$.
	\end{itemize}
\end{proposition}

We show that the periodicity of outcomes and values still holds if \Subs\ rulesets are composed using the push-the-button construction.

\begin{theorem}
	\label{th:subtraction2}
	Suppose that $\R_2$ is a ruleset over $\mathbb N$ with periodic \P-positions of period $k$, and $\R_1= \Subs(S)$ for a finite set $S$. The ruleset $\R_1 \swComp \R_2$ has ultimately periodic \P-positions. The lengths of the period and pre-period are at most $k\cdot2^{\max(S)}$.
	
	In addition, if the values of $\R_2$ are also periodic of period $k$, then the values of $\R_1 \swComp \R_2$ are periodic with period and pre-period at most $k(|S|+2)^{\max(S)}$.
\end{theorem}

\begin{proof}
	We note $M = \max(S)$. Remark that the function $n \mapsto o_\R(n)$ is such that its value on an integer $n$ only depends on two things:
	\begin{itemize}
		\item The value of $o_{\R_2}(n)$, which only depends on the congruence class of $n$ modulo $k$.
		\item The value of $o_\R(n-i)$ for $1 \leq i \leq M$.
	\end{itemize}
	As a consequence, we can write:
	$$ o_\R(n) = f(o_\R(n-1), \ldots, o_\R(n-M), n\bmod k)$$
	for some function $f: \{\P, \N\}^{M} \times \mathbb Z_k \longrightarrow \{\P, \N \} $.
	If we note $X$ the set  $\{\P, \N\}^{M} \times \mathbb Z_k$, we can define the function $g: X \rightarrow X$  as:
	\begin{align*}
	g(a_1, \ldots a_{M}, m) = (f(a_1, \ldots a_{M}, m), a_1, \ldots a_{M-1}, m+1 \bmod k).
	\end{align*}
	We denote by $F_0$ the initial vector $(o_\R(M-1), \ldots o_\R(0), M)$, and $g^{(m)}$ the function $g$ composed $m$ times. We can see that the first element of the tuple $ g^{(m)}(F_0) $ is $o_\R(M+m-1)$. Now, since $|X| = k2^M$, there exists $n_0$ and $p$ smaller than $|X|$ such that $g^{(n_0+p)}(F_0) = g^{(n_0)}(F_0)$, and consequently, for all $m \geq n_0$, $o_\R(m+p) = o_\R(m)$.
	
	For the periodicity of the Grundy function, simply observe that there are at most $|S|+1$ possible moves at any moment. Consequently the Grundy function takes values in the set $\{ 0, \ldots, |S|+1\}$. Using exactly the same argument as above with the set $X= \{ 0, \ldots, |S|+1\}^M \times \mathbb Z_k$, we get the desired bound on the length of the period.
\end{proof}

For some specific subtraction sets, it is possible to describe exactly the period of the \P-positions:

\begin{theorem}
\label{th:subtraction1}
Let $k_1, k_2$ be two positive integers, and consider the ruleset $\R = (\R_1 \swComp \R_2)$, with $\R_1 = \Subs(\{ 1,\ldots,k_1\})$ and $\R_2 = \Subs(\{ 1,\ldots, k_2\})$.

Let $a$ be the smallest integer such that $(k_1+1)a \equiv -1 \bmod (k_2+1) $. The \P-positions of $\R_1 \swComp \R_2$ are periodic with period $(k_1+1)a + 1$. If no such $a$ exist, the period is simply $k_1+1$.
\end{theorem}

\begin{proof}
	Assume that there exists an $a$ such that $(k_1+1)a \equiv -1 \bmod (k_2+1) $, and take the smallest one with this property. Again, the outcome of $(m)_\R$ is determined by two things:
	\begin{itemize}
		\item The outcome of $(m)_{\R_2}$, the game obtained after pushing the button, which only depends on the value of $m$ modulo $(k_2 +1)$.
		\item The outcome of the $k_1$ previous positions $m-1, m-2, \ldots m - k_{1}$, and more precisely whether these positions are all \N-positions for \R\ or not.
	\end{itemize}
	 As a consequence, if we manage to find two positions $m$ and $m'$ such that ${m \equiv m' \bmod (k_2 + 1)}$ and $o_\R(m)= o_\R(m') = \P$, then the \P-positions will be periodic with period $m' - m$.

	Now, we can check easily that $o_\R(0) = \N$ and $o_{\R}(1) = \P$. The games $(2)_\R, (3)_\R \ldots (k_1 + 1)_\R$ have outcome \N\ since they all have a move to $(1)_\R$. The position $(k_1+2)$, and similarly all positions $(b(k_1+1) + 1)$ for $b < a$ are \P-positions for ruleset $\R$ and the positions in-between are \N-position. Indeed by minimality of $a$, pushing the button on a position $b(k_1+1) +1$ for $b < a$ is a losing move. Furthermore, there is no move from $b(k_1+1) +1$ to some $b'(k_1+1) +1$, and finally every position in-between these position has a move to some $b(k_1 +1) +1$ for some $b \geq 0$. The position $a(k_1+1) +1$ is a winning position since pushing the button is a winning move by definition of $a$. Consequently, $(a(k_1+1) + 2)_\R$ is a losing game: all $k_1$ previous positions are \N-positions, and pushing the button is a losing move. By definition of~$a$, we have $a(k_1+1) + 2 \equiv 1 \bmod (k_2+1)$, and $o_\R(1) = \P$. Consequently the sequence of outcomes is periodic with period  $a(k_1+1) + 1$.

	If such a $a$ does not exists, according to Proposition~\ref{prop:subsk} we have:
	\begin{itemize}
		\item The \P-positions of $\R_1$ in misere convention are the positions $m$ such that $m \equiv 1 \bmod (k_1+1)$.
		\item The \P-positions of $\R_2$ are the positions $m$ for which ${m \equiv 0 \bmod (k_2+1)}$.
	\end{itemize}
	By hypothesis on $k_1, k_2$, these two sets do not intersect. Additionally, $0$ is the only terminal position of $\R_1$ and it is a \P-position for $\R_2$. By Proposition~\ref{prop:push}, the \P-positions of $\R$ are the \P-positions of $\R_1$ in misere convention which are periodic with period $k_1+1$.

\end{proof}

\section{Push compounds of placement games: the \textsc{Cram} ruleset}
\label{sec:cram}

In previous sections, we studied several games played with heaps of tokens. We now move to a different kind of game: placement games. We look in particular at the game \textsc{Domineering} and its impartial variant \textsc{Cram}, and study a push compound of this game.

\begin{definition}
The game of \textsc{Cram} is played on a square grid. A move consists in placing either a vertical or an horizontal $2 \times 1$ domino on the grid, not overlapping other dominoes.
\end{definition}

Using some symmetry arguments, this game is known to be solved on rectangular grids with at least one even dimension. On even by even grids, the second player has a winning strategy by playing his opponent's central symmetric move. On even by odd boards, the first player can place a domino in the center, and then apply the central symmetric strategy. For $(2n+1)\times (2m+1)$ grids however, the computational complexity of finding the outcome is an open problem.\\

We here study a variation of {\sc Cram} where the players first play only vertical dominoes, then, when the button is pushed, switch to horizontal ones (and since the rules are defined on the same game board they are compatible). We call \textsc{Push Cram} this new game, and present a solution for some of its non trivial sub-cases.
In this subsection, the notation $m \times n$ will denote the grid with $m$ rows and $n$ columns.

The game of \textsc{Cram} played on a $1 \times n$ row has the same rules as the octal game $0.07$. Its values are periodic, with period $34$ and pre-period $52$~\cite{winningWays}. In particular, if we have $\G_{0.07}(n) = 0$, then $n$ is odd (i.e.: positions $1 \times (2k)$ have outcome \N). In \textsc{Push Cram} played on an $m \times n$ grid, when the button is pushed, the game decomposes into a disjunctive sum of $m$ distinct $0.07$ positions. Indeed, after the button is pushed the players can only place horizontal dominoes, and each domino placed on a row does not affect the other rows. In this case, the outcome can be computed easily using Proposition~\ref{prop:grundy} and the periodicity of the values of $0.07$.

\begin{proposition}
\label{prop:cram}
Let $k$ and $n$ be two non-negative integers. We have the following outcomes for \textsc{Push Cram}:
\begin{enumerate}[label=(\roman{enumi}), topsep=0.2em, itemsep=-0.2em]
\item $(2k) \times n$ is an \N-position,
\item $(2k+1) \times n$ is an \N-position if $\G_{0.07}(n) = 0$,
\item $3\times (2k)$ is a \P-position,
\item $n \times 3$ is a \P-position if and only if $\mathcal G_{0.07} (n) = 0$,
\item $(2k+1) \times 4$ is a \P-position.
\end{enumerate}
\end{proposition}

\begin{proof}
	For each of these cases, we describe a winning strategy for one of the players.
\begin{enumerate}[label=$(\roman{enumi})$, topsep=0.2em, itemsep=-0.2em, wide, labelindent=3pt]
\item The first player directly wins by pushing the button. Indeed, once the button is pushed, the game decomposes into a disjunctive sum of $2k$ games, each game corresponding to a row. Since all rows have the same length, the value (obtained by the XOR of the values of each game) is zero.
\item Again, the first player pushes the button and wins since once the button is pushed, the games corresponding to each row have value~$0$.

\item We give a strategy for the second player: play the move symmetrically to the first player relatively to the vertical line cutting the grid in half. Since such a move is always possible, we only have to show that for the first player pushing the button is always a losing move. Each move of the players cuts the horizontal rows. Looking at the connected horizontal pieces of rows, they go in pairs with their symmetric except for three of them in the center which are their own symmetric. Among these three, two of them have the same length, and the remaining one has even length (see Figure \ref{fig:case3lines}). Moreover, $0.07$ played on a row of even length has a value different from zero. As a consequence, for the first player pushing the button on such a symmetric game is losing.

\begin{figure}[!ht]
	\centering
	\ifNB
\else
	\includegraphics[width=.5\textwidth]{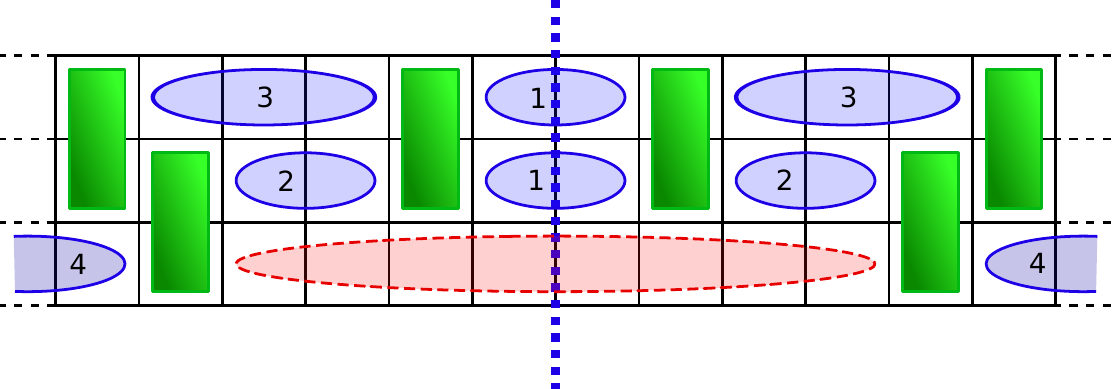}
	\caption{\label{fig:case3lines} Case $(iii)$. The blue dotted vertical line is the symmetry axis All the rows go by pair except one (in red with a dashed border) which has even size. Since $0.07$ played on a row of even size has a non-zero value, pushing the button is a losing move on symmetric positions.}
	\fi
\end{figure}

\item  The second player can play his optimal strategy for $0.07$ on the last column and play symmetrically on the two other ones. At some point, the opponent will run out of moves and will be forced to push the button, giving a winning position to the second player (see Figure~\ref{fig:cases1}). At any point during the game, when it's the first player's turn, there is an odd number of rows with $2$ or $3$ empty cells. Indeed, if the first player plays on one of the two first columns, by playing symmetrically, two rows with either $2$ or $3$ empty cells are transformed into rows with $0$ or $1$ empty cells. Additionally, playing into the last column changes the number of free cells of the two rows from $3$ to $2$ or from $1$ to $0$. After pushing the button, the value of the game is the number modulo $2$ of rows with $2$ or $3$ empty cells. Consequently, for the first player, pushing the button is a losing move.

\begin{figure}[!hb]
	\centering
	\ifNB
\else
	\includegraphics[width=.9\textwidth]{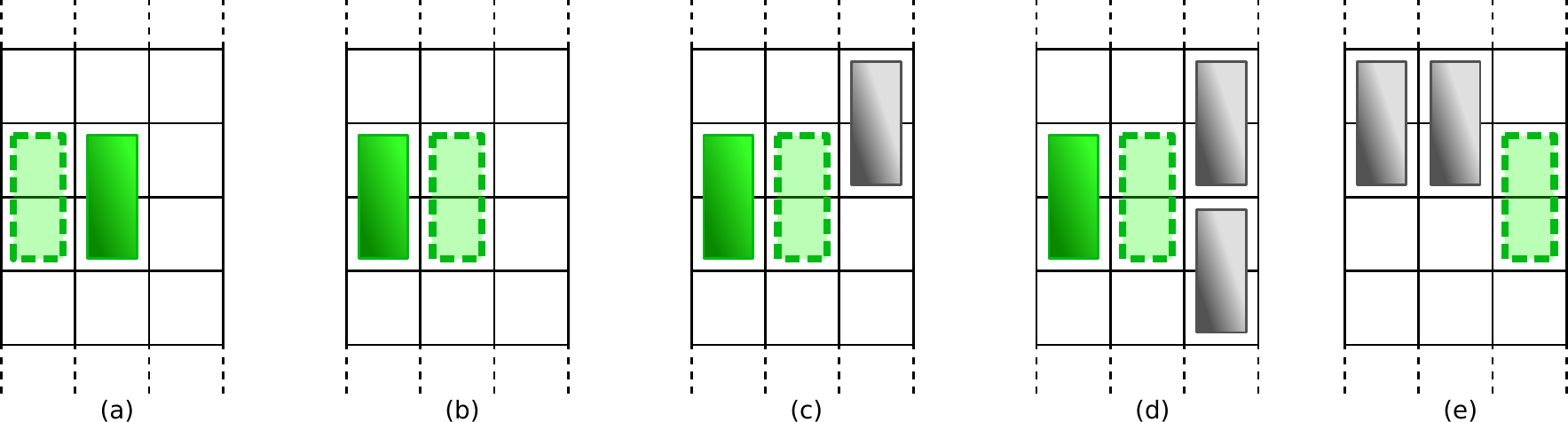}
	\caption{\label{fig:cases1} Subcases  $(a) - (d)$ consist in playing symmetric moves on the two first columns. $(e)$ is the case where the player plays on the last column. Grey tiles correspond to previous moves. The green tile is the last move played by the opponent, and the dotted one is the move that the player wants to play by applying the strategy. }
	\fi
\end{figure}

\item Group the columns in two consecutive pairs. Either the second player can push the button and win, or, by playing the symmetric in the paired column, the value of the game for $0.07$ does not change (see Figure~\ref{fig:cases2}). If he plays on a side column (case (a)), his move removes two rows of size $1$, and does not change the value for $0.07$ on these rows. If he plays on a middle column (case (b)), then he changes rows of length $1$ to $0$, and rows of length $3$ to $2$. Since $\G_{0.07}(1) = \G_{0.07}(0)$ and $\G_{0.07}(3) = \G_{0.07}(2)$, the value for $0.07$ is not changed either by the move.

\begin{figure}[!ht]
\centering
\ifNB
\else
\includegraphics[width=.4\textwidth]{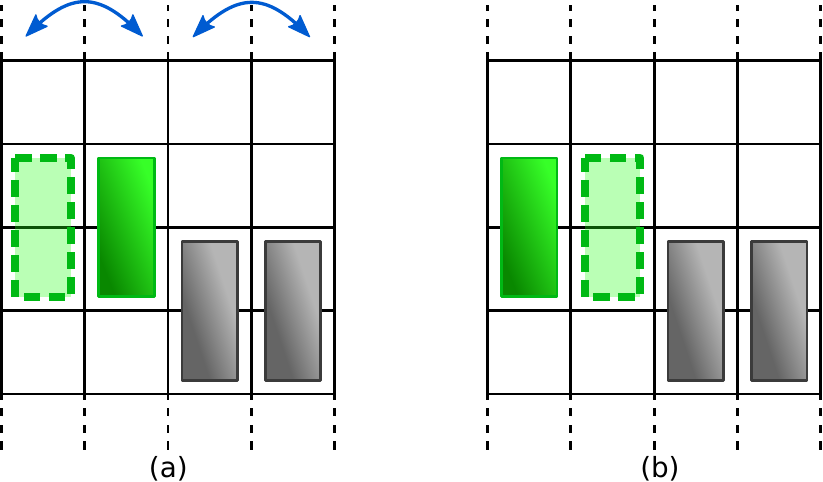}
\caption{\label{fig:cases2} The second player can play symmetrically to the other player on each pair of column.}
\fi
\end{figure}
\end{enumerate}
\end{proof}

A game is a \emph{bluff} game \cite{bluff} if it is a first player win, and any first move is a winning move.
Using numeric computations, we were able to formulate the following conjecture:

\begin{conjecture}
For all $k \geq 0$, the grid of size $3 \times (2k+1)$ is an \N-position for the game \textsc{Push-Cram}. In addition, it is a bluff game.
\end{conjecture}

This conjecture was verified up to grids of size $3 \times 25$ for the outcome, and $3\times 13$ for the bluff property. The remaining unsolved case for \textsc{Push Cram} on rectangular grids is when there is an odd number of rows, and a number of columns $m$ such that $\G_{0.07}(m) \neq 0$. The complexity of computing the outcome of a given position is also unknown.

\section{Open problems and conclusion}

Other variants of the switch procedure could be interesting. In this study we took the convention that pushing the button requires a move. It is possible to look at variants for which it does not, and the players can push the button either before playing, after playing, or both. It is also possible to study the case where pushing the button cannot be done if there is no move left for the first ruleset. With this convention, the push operator is a direct extension of the pass move construction: the players are allowed a single pass move which can be played once, at any moment during the game except at the very end when there is no other move available. This construction was studied in \cite{passMove}, and also in \cite{nimpass} for the particular case of \textsc{Nim}, and in \cite{octalpass} for \textsc{Octal} rulesets with a pass. For all these variations of the switch construction, it is possible to get results similar to Proposition~\ref{prop:push} giving sufficient conditions for which the losing positions are not modified. It is also likely that Theorem~\ref{th:subtraction2} on the periodicity of \textsc{Subtraction} games can also be adapted for these variations, perhaps with a different bound on the period.

The construction might also be interesting to look at from a complexity point of view. Given two compatible rulesets $\R_1$ and $\R_2$, are there relations between the complexities of computing the outcome of positions of $\R_1, \R_2$ and $\R_1 \swComp \R_2$?

The game \textsc{Push Cram} is a special case of a more general construction that allows to transform any partizan game into a push ruleset: the player both start by playing only Left-type moves, then when the button is pushed, they continue only using Right's moves. \textsc{Push Cram} is the game obtained by applying this construction to \textsc{Domineering}. It could be interesting to study how other partizan games are modified with this construction.

It seems also relevant to look at games that can be decomposed into a sum of positions. Typically, this happens for \textsc{Cram}, when the board can be split into several connected components. Usually, the outcome of  such games is obtained by computing the values of the components. Yet, for switch rulesets (\textsc{Push Cram} for example), these positions do not correspond to a disjunctive sum. Indeed, pushing the button changes the rules for all the components at the same time. In an upcoming work \cite{futurework}, we look at how to handle sums of positions for push rulesets using tools from misere game theory.

\bibliographystyle{alpha}
\bibliography{bib}


\newpage
\begin{appendices}
	
\section{\Rthl}

\recCaract*
\begin{proof}
	By definition of $A_n$ and $B_n$, the sets $\{A_n\}_{n\geq 0}$ and $\{B_n\}_{n\geq 0}$ are complementary. Consequently there is no move from $(A_n, B_n)$ to some $(A_m, B_m)$. Moreover, by definition we also have $\frac{B_n}{A_n} > \Phi$, hence pushing the button is a losing move.
	
	Let $(a,b)$ be a position with $a \leq b$. One of the following holds:
	\begin{itemize}
		\item either $\frac{b}{a} < \Phi$ in which case the first player wins by pushing the button (see Proposition~\ref{euclidPos}).
		\item or $\frac{b}{a} > \Phi$. Since $\{A_n\}_{n \geq 0}$ and $\{B_n\}_{n \geq 0}$ are complementary sequences, there is an integer $n$ such that either $a = A_n$ or $a = B_n$. In the first case, since $\frac{b}{a} > \Phi$, we have $b> B_n$ and the move to $(A_n, B_n)$ is a winning move. In the second case, $b \geq a = B_n > A_n$, hence there is a move to $(B_n, A_n)$.
	\end{itemize}
\end{proof}

\fiboCaract*

\begin{proof}
	This result is the consequence of the following two facts:
	\begin{itemize}
		\item the pair $(a,b)$ is a Wythoff pair if and only if the Fibonacci representation $s_a$ of $a$ ends with an even number of $0$, and the representation of $b$ is $s_a0$ (see \cite{WythoffFraenkel} for a proof of this result).
		
		\item The Fibonacci representation of $F_{2n+1} -1$ is $(10)^n1$. Indeed, we have the following:
		\begin{align*}
		\sum_{k=1}^n F_{2k} &= \sum_{k=1}^n (F_{2k} + F_{2k-1}) - \sum_{k=1}^n F_{2k-1} \\
		&= \sum_{k=1}^n F_{2k+1} - \sum_{k=1}^n F_{2k-1} \\
		&= F_{2n+1} - F_1 = F_{2n+1} - 1.
		\end{align*}
		
		A similar argument shows that the Fibonacci representation of ${F_{2n+2}-1}$ is $(10)^n$.
	\end{itemize}
	
	Using these facts, the statement of the Theorem is simply a consequence of Theorem \ref{rthl1}. Given $(A_n, B_n)$ as in the statement of the Theorem, either $(A_n,B_n)$ is a Wythoff pair with $A_n \neq F_{2k+1}-1$, or $A_n = F_{2k+2}-1$ and $B_n = F_{2k+3}-1$. In the first case the Fibonacci representation $s_A$ of $A_n$ ends with an even number of zeros, and $B_n = s_A 0$, with $s_A \neq (01)^k$. In the second case, the Fibonacci representations of $A_n$ and $B_n$ are respectively $(10)^{k}$ and $(10)^k1$.
\end{proof}

\ineqPpos*
\begin{proof}
	The left inequality comes from the fact that if $c < \lceil \Phi b \rceil$, then $\frac c b  < \Phi$, hence moving to $(0,b,c)$ is winning.
	Using the characterization of Theorem~\ref{th:rthlrec}, the right inequality holds when $a=1$. Assume $a > 1$, and suppose by contradiction that there is at least one \P-position $(a,b,c)$ with $a \leq b \leq c$ such that $c \geq \lceil \Phi b \rceil + a$. We look at the smallest such position $(a,b,c)$ for the lexicographic order.
	
	By assumption, there exists $c' < c$ such that: $c \equiv c' \bmod a$, and $ \lceil \Phi b \rceil \leq c' \leq \lceil \Phi b \rceil + a-1$. Since there is a move from $(a,b,c)$ to $(a,b,c')$ and $(a,b,c)$ is a \P-position, $(a,b,c')$ is a \N-position. Because $(a,b,c')$ is a \N-position, there is a winning move from $(a,b,c')$ to some \P-position. From $(a,b,c')$, there are only $3$ possible types of moves:
	\begin{itemize}
		\item The move to $(0,b,c')$ is losing since by assumption on $c'$, $c' \geq \lceil \Phi b\rceil $.
		\item  Any move to $(a,b,c'')$ with $c'' < c'$ is also losing since otherwise we would have a winning move from $(a,b,c)$ to $(a,b,c'')$ which is a contradiction of the fact that $(a,b,c)$ is a \P-position.
		\item A move to a position $(a,b-qa,c')$ for some $q>0$.
	\end{itemize}
	In this third case, there are two possibilities. Either $b - qa \geq a$ in which case we have:
	$$ c' \geq \lceil \Phi b\rceil \geq \lceil \Phi (b-qa) \rceil + \lfloor \Phi q a\rfloor > \lceil \Phi (b-qa) \rceil + a-1.$$
	This is a contradiction of the minimality of $(a,b,c)$. The other possibility is $b - qa < a$. Using again the minimality of $(a,b,c)$, we know that $(b-qa, a, c')$ satisfies the inequalities of the theorem, hence:
	
	\begin{align*}
	c' \leq \lceil \Phi a \rceil + b-qa -1 < \lceil \Phi a \rceil + b - a.
	\end{align*}
	But we already know that $c' \geq \lceil \Phi b \rceil$, which implies:
	\begin{align*}
	\lceil \Phi b \rceil - b \leq c' - b < \lceil \Phi a \rceil - a.
	\end{align*}
	This is a contradiction of the hypothesis $b \geq a$.
\end{proof}

\moduloPpos*
\begin{proof}
	Suppose by contradiction that there exists $m$ such that there is no \P-position $(a,b,c)$ with $c \equiv m \bmod a$. By taking the unique $c$ such that $\lceil \Phi b \rceil \leq  c \leq \lceil \Phi b \rceil + a-1$, the position $(a,b,c)$ is an \N-position. For this position, the move to $(0,b,c)$ is a losing move. Moreover, there is no winning move to a $(a,b,c')$ by hypothesis ($c$ and $c'$ being in the same congruence class modulo $a$), and using the same kind of argument as in the proof of Theorem~\ref{th:ineqPpos}, we can show that there is no possible move to a \P-position of the form $(a,b-qa,c)$.
\end{proof}
	
\end{appendices}

\end{document}